\documentclass[10pt]{article}
\usepackage[english]{babel}
\usepackage{array,amsmath,url}
\usepackage{graphicx}
\usepackage[english]{babel}
\usepackage{xcolor,amssymb}

\def\openone{\leavevmode\hbox{\small1\kern-6.3pt\normalsize1}}

\def\im{\mbox{Im \,}}
\def\re{\mbox{Re \,}}

\newtheorem{corollary}{Corollary}
\newtheorem{proof}{Proof}
\newtheorem{remark}{Remark}
\newtheorem{proposition}{Proposition}

\begin{document}

\allowdisplaybreaks

\arraycolsep=2pt


\title{On Asymptotic Dynamical Regimes of Manakov $N$-soliton Trains in Adiabatic Approximation}


\author{Vladimir S. Gerdjikov$^{1}$, Michail D. Todorov$^{2}$\\[1.mm]
\small\it $^1$National Research Nuclear University MEPHI, 115409 Moscow, Russian Federation \\
\small\it Institute of Mathematics and Informatics  Bulgarian Academy of Sciences,   1113 Sofia, Bulgaria \\
\small\it Institute for Advanced Physical Studies, New Bulgarian University, 1618 Sofia, Bulgaria, \\
\small\it email: vgerdjikov@math.bas.bg \\[2.mm]
\small\it $^2$Dept of Mathematics and Statistics, San Diego State University, 92182-0005 San Diego, CA, USA \\
\small\it Faculty of Applied Math. and Informatics,  Technical University of Sofia, 1000 Sofia,  Bulgaria, \\
\small\it e-mail: mtod@tu-sofia.bg}
\date{}

\maketitle

\begin{abstract}
We analyze the dynamical behavior of the $N$-soliton train in the adiabatic approximation of the  Manakov model.
The evolution of Manakov $N$-soliton trains is described by the complex Toda chain
(CTC) which is a completely integrable dynamical model.
Calculating the eigenvalues of its Lax matrix allows us to determine the asymptotic velocity of each soliton. So we describe sets of soliton parameters that ensure one of the two main types of asymptotic regimes: the bound state regime (BSR) and the free asymptotic regime (FAR).
In particular we find explicit description of special symmetric configurations of $N$ solitons that ensure BSR and FAR.
We find excellent matches  between the trajectories of the solitons predicted by CTC  with the ones calculated numerically from the Manakov system for wide classes of soliton parameters.
This confirms the validity of our model.\medskip

\end{abstract}

\section{Introduction}

The solitons and their interactions find numerous applications of  in many areas of today nonlinear physics, such as hydrodynamics, nonlinear optics, Bose-Einstein condensates, etc. \cite{Manak, Agr, KivMalo, AndLis, UGL,  ABDK, carretero}. This explains why it is important to study their interactions. The first results on soliton interactions were obtained by Zakharov and Shabat \cite{ZaSha, ZMNP}. There they proved that the nonlinear Schr\"odinger equation
\begin{equation}\label{eq:nls}\begin{split}
iu_t + \frac{1}{2} u_{xx} +|u|^2 u(x,t)=0.
\end{split}\end{equation}
can be integrated by the inverse scattering method (ISM). Then they constructed the $N$-soliton solution of (\ref{eq:nls}) and calculated their limits for $t\to \infty$ and  $t\to -\infty$, assuming that all solitons have different velocities. Comparing the asymptotics they concluded that the soliton interactions are purely elastic, i.e., no new solitons can be created. In addition the solitons preserve their amplitudes and velocities, and the only effect of the interactions are relative shifts of the center of masses and phases.

Later Karpman and Solov'ev proposed another approach to the soliton interactions based on the
adiabatic approximation \cite{KarpS, Karp79}. They proposed to model the $N$-soliton trains of the  NLS eq. (\ref{eq:nls}). By $N$-soliton train they meant  a solution of the NLS eq. with initial condition:
\begin{equation}\label{eq:Nst}\begin{aligned}
u(x,t=0) &= \sum_{k=1}^N \vec{u}_k(x,t=0), &\qquad
u_k(x,t) &= \frac{ 2\nu_k e^{i\phi_k}}{\cosh(z_k)}  ,\\
z_k &= 2\nu_k (x-\xi_k(t)), &\qquad \xi_k(t) &=2\mu_k t +\xi_{k,0}, \\
\phi_k &= \frac{ \mu_k }{\nu_k} z_k + \delta_k(t), &\qquad
\delta_k(t) &=2(\mu_k^2+\nu_k^2) t +\delta_{k,0}.
\end{aligned}\end{equation}
The adiabatic approximation holds true if the soliton parameters  satisfy \cite{KarpS}:
\begin{eqnarray}\label{eq:ad-ap}
&& |\nu _k-\nu _0| \ll \nu _0, \quad |\mu _k-\mu _0| \ll \mu _0,
\quad |\nu _k-\nu _0| |\xi_{k+1,0}-\xi_{k,0}| \gg 1,
\end{eqnarray}
where $\nu _0 = \frac{ 1}{N }\sum_{k=1}^{N}\nu _k$, and $ \mu _0 = \frac{1}{N }\sum_{k=1}^{N}\mu _k$ are the average amplitude and velocity respectively. In fact we have two different scales:
\[ |\nu _k-\nu _0| \simeq \varepsilon_0^{1/2}, \qquad
|\mu _k-\mu _0| \simeq \varepsilon_0^{1/2}, \qquad
|\xi_{k+1,0}-\xi_{k,0}| \simeq \varepsilon_0^{-1}.
\]
In this approximation the  dynamics of the $N$-soliton train is described by
a dynamical system for the $4N$ soliton parameters. What Karpman and Solov'ev did was to derive the dynamical system for the two soliton interactions: a system of 8 equations for the 8 soliton parameters. They were able also to solve it analytically.

Later their results were generalized to $N$-soliton trains \cite{PRL77, PRE55, PLA98}. The corresponding model can be written down in the form :
\begin{equation}\label{eq:nuk0}\begin{aligned}
\frac{ d\lambda _k }{ dt} &= -4\nu _0 \left(e^{Q_{k+1}-Q_k} - e^{Q_k -Q_{k-1}} \right) , \\
\frac{ d Q_k }{ dt} &=  - 4\nu_0 \lambda_k ,
\end{aligned}\end{equation}
where $\lambda _k=\mu _k+i\nu _k $ and
\begin{equation}\label{eq:140.1}\begin{aligned}
Q_k &= -2\nu _0\xi_k + k \ln 4\nu _0^2 - i (\delta _k+\delta _0 + k\pi -2\mu _0 \xi_k), \\
\nu _0 &= \frac{1}{N} \sum_{s=1}^{N} \nu _s, \qquad \mu _0 = \frac{1}{N} \sum_{s=1}^{N} \mu _s, \qquad \delta _0 = \frac{1}{N} \sum_{s=1}^{N} \delta _s.
\end{aligned}\end{equation}
Obviously the system (\ref{eq:nuk0}) becomes the Toda chain with free ends for
the complex variables $Q_k$:
\begin{equation}\label{eq:ctc}\begin{split}
\frac{ d^2 Q_k }{ dt^2} &=  - 4\nu_0 \frac{d \lambda_k}{d t}= 16\nu _0^2 \left(e^{Q_{k+1}-Q_k} - e^{Q_k -Q_{k-1}} \right) \qquad k=2,\dots, N-1  ,\\
\frac{ d^2 Q_1 }{ dt^2} &=  16\nu _0^2 e^{Q_{2}-Q_1} , \qquad \frac{ d^2 Q_N }{ dt^2} =  -16\nu _0^2 e^{Q_{N}-Q_{N-1}} .
\end{split}\end{equation}
which is known as the complex Toda chain (CTC).

It is well known that the standard (real) Toda chain is an integrable system  \cite{Moser, Manak, Fla}. In the case of (\ref{eq:ctc}), which is known as Toda chain with open ends, it was possible to write down its solutions explicitly \cite{Moser}. An important fact is that these solutions depend analytically on their parameters and can be easily generalized to the CTC.

In fact some time ago a special configurations of soliton trains that are modeled by the real Toda chain \cite{Arn, Arn2}. To this end we must choose solitons with equal amplitudes (i.e., $ \nu_k=\nu_0$), vanishing initial velocities ($\mu_k =0$, and out-of phase $\delta_{k+1}-\delta_k=\pi$. It is easy to see that under these assumptions $Q_k$ become real valued and (\ref{eq:ctc}) become the standard Toda chain.

The adiabatic approach of Karpman and Solov'ev has a drawback: it is an approximate method whose precision is determined by $\varepsilon_0$. On the other hand it has the advantages: first, it is not limited only to solitons with different velocities, and second, it can take into account possible perturbations of the NLS \cite{PRL77, PRE55, PLA98}.

Another important generalization of the NLS equation is known as the Manakov model \cite{Manak} (vector NLS):
\begin{equation}\label{eq:vnls}\begin{split}
i\vec{u}_t + \frac{1}{2} \vec{u}_{xx} +(\vec{u}^\dag, \vec{u}) \vec{u}(x,t)=0.
\end{split}\end{equation}
The corresponding vector $N$-soliton train is determined by the initial condition:
\begin{equation}\label{eq:Nstr}\begin{aligned}
\vec{u}(x,t=0) &= \sum_{k=1}^N \vec{u}_k(x,t=0), &\qquad
\vec{u}_k(x,t) &= {2\nu_k e^{i\phi_k}\over \cosh(z_k)} \vec{n}_k , \\
z_k &= 2\nu_k (x-\xi_k(t)), &\qquad \xi_k(t) &=2\mu_k t +\xi_{k,0}, \\
\phi_k &= \frac{ \mu_k }{ \nu_k} z_k + \delta_k(t), &\qquad
\delta_k(t)&=2(\mu_k^2+\nu_k^2) t +\delta_{k,0},
\end{aligned}\end{equation}
where the constant polarization vector $\vec{n}_k$ is normalized by
\[ \vec{n}_k = \left(\begin{array}{c} \cos(\theta_k) e^{i\gamma_k} \\ \sin(\theta_k) e^{-i\gamma_k}  \end{array}\right), \qquad
(\vec{n}_k ^\dag , \vec{n}_k)= 1, \qquad \sum_{s=1}^{n} \arg \vec{n}_{k;s} =0.\]
Therefore each Manakov soliton is parametrized by 6 parameters.

It was natural to extend the Karpman-Solov'ev method to the Manakov model. The result is known as
the generalized CTC (GCTC)  \cite{VG98, VG03, GDM07, BJP38}. Of course later the GCTC was also adapted to treat the effects of several types of perturbations on solitons \cite{GBS05, MTVGAK1, MTVGAK2, MiWaF, SchDok, DokSch}.

The advantage of the integrability of the CTC and GCTC is in the fact that knowing the initial set of soliton parameters one can predict the asymptotic regime of the soliton train \cite{PRL77, PRE55, PLA98}. On the other hand it is possible to find the set of constraints on the soliton parameters that would ensure given asymptotic regime. These constraints were derived and analyzed for $2$ and $3$-soliton trains; for larger number of solitons only fragmentary results such as the quasi-equidistant propagation of solitons \cite{PLA98}.

The aim of the present paper is to reinvestigate these results and to demonstrate several configuration of multisoliton trains for which one can predict that they will go into bound state regime (BSR) or into free asymptotic regime (FAR).
In Section 2 we outline the derivation  of the GCTC model, see eq.  (\ref{eq:Vnuk0}) below which
now depends also on the polarization vectors $\vec{n}_k$ and models the behavior of the
$N$-soliton train of the vector NLS. We also formulate the Lax representation for the GCTC and explain how it can be used to determine the asymptotic regime of the soliton train.
In Section 3 we formulate two classes of explicit constraints on the soliton parameters that are responsible for BSR and FAR. The first class are generic conditions that ensure that the Lax matrix becomes either real or purely imaginary. The second class are based on special explicit constraints on the soliton parameters that make the eigenvalues of the Lax matrix proportional to each and easier to establish if they are real ir purely imaginary.

\section{Preliminaries}
\subsection{Variational Approach and Generalized CTC  }

The Lagrangian of the vector NLS perturbed by external potential is:
\begin{equation}\label{eq:Lag}\begin{aligned}
\mathcal{L}[\vec{u}] &= \int_{-\infty}^{\infty} dt\;  {i\over 2} \left[
(\vec{u}\dag ,\vec{u}_t) -(\vec{u}_t^\dag,\vec{u})  \right] - H,\qquad
H[\vec{u}] &=\int_{-\infty}^{\infty} dx\;  \left[ -{1\over 2} (\vec{u}_x^\dag,\vec{u}_x)
+ \frac{1}{2} (\vec{u}^\dag,\vec{u})^2  \right].
\end{aligned}\end{equation}
Then the Lagrange equations of motion:
\begin{eqnarray}\label{eq:120.3}
{d\over dt} \frac{ \delta \mathcal{L} }{ \delta \vec{u}_t^\dag} -
\frac{ \delta \mathcal{L} }{ \delta \vec{u}\dag} =0,
\end{eqnarray}
coincide with the vector NLS with external potential $V(x)$.

Next we insert $\vec{u}(x,t) = \sum_{k=1}^{N} \vec{u}_k(x,t)$ (see eq. (\ref{eq:Nstr})) and integrate
over $x$ neglecting all terms of order $\epsilon$ and higher. In doing this we assume that
$\xi_1<\xi_2 < \cdots < \xi_N$ at $t=0$ and use the fact, that only the nearest neighbor solitons
will contribute. All integrals of the form:
\begin{equation}\label{eq:int1}\begin{split}
\int_{-\infty}^{\infty} dx\; (\vec{u}_{k,x}^\dag,\vec{u}_{p,x}), \qquad
\int_{-\infty}^{\infty} dx\;  (\vec{u}_k^\dag,\vec{u}_p ),
\end{split}\end{equation}
with $|p-k| \geq 2$ can be neglected. The same holds true also for the integrals
\[ \int_{-\infty}^{\infty} dx\;  (\vec{u}_k^\dag,\vec{u}_p) (\vec{u}_s^\dag,\vec{u}_l), \]
where at least three of the indices $k,p,s,l$ have different values.
In doing this  key role play  the following integrals:
\begin{equation}\label{eq:J}\begin{aligned}
{\cal   J}_{2}(a) &=\int_{-\infty}^{\infty} \frac{dz \,  e^{iaz} }{ 2 \cosh^2 z} =  
\frac{ \pi a }{ 2\sinh \frac{ a\pi }{ 2 } } ,\\
K(a,\Delta )&\equiv  \int_{-\infty }^{\infty } 
\frac{ dz\; e^{iaz}  }{ 2\cosh z \cosh (z+\Delta ) } =  
\frac{ \pi (1-e^{-ia\Delta }) }{ 2i \sinh(\Delta )\sinh(\pi a/2) },
\end{aligned}\end{equation}

Thus  after long calculations we obtain:
\begin{equation}\label{eq:139.1}
\begin{aligned}
\mathcal{L} &= \sum_{k=1}^{N} \mathcal{L}_k + \sum_{k=1}^{N}
\sum_{n=k\pm 1} \widetilde{\mathcal{L}}_{k,n}, &\quad
\mathcal{L}_{k,n} &= 16\nu _0^3 e^{-\Delta_{k,n}}(R_{k,n}+R_{k,n}^*), \\
R_{k,n}&= e^{i(\widetilde{\delta} _n- \widetilde{\delta} _k)}
(\vec{n}_k^\dag \vec{n}_n), &\quad   \widetilde{\delta }_k &=\delta _k- 2\mu _0\xi_k, \\
 \Delta _{k,n} &=2s_{k,n}\nu _0(\xi_k -\xi_n)\gg 1, &\quad
s_{k,k+1} &=-1, \qquad s_{k,k-1}=1,
\end{aligned}
\end{equation}
where
\begin{equation}\label{eq:123.2}\begin{split}
\mathcal{L}_k &=  -2i\nu _k \left(
(\vec{n}_{k,t}^\dag,\vec{n}_k)-(\vec{n}_k^\dag,\vec{n}_{k,t})
\right) +8\mu _k\nu _k \frac{ d\xi_k }{dt }  - 4\nu _k 
\frac{ d\delta _k }{ dt }
-8\mu _k^2\nu _k + \frac{  8\nu _k^3 }{ 3 }
\end{split}\end{equation}

The equations of motion are given by:
\begin{eqnarray}\label{eq:eq_m}
\frac{ d  }{ dt } \frac{ \delta \mathcal{L} }{\delta p_{k,t}} -
\frac{ \delta \mathcal{L}  }{\delta p_k} =0,
\end{eqnarray}
where $p_k $ stands for one of the soliton parameters: $\delta _k $,
$\xi_k $, $\mu _k $, $\nu _k $ and $\vec{n}^\dag_k $.
The corresponding system  is a generalization of CTC:
\begin{equation}\label{eq:Vnuk0}\begin{aligned}
\frac{ d\lambda _k }{dt} &= -4\nu _0 \left(e^{Q_{k+1}-Q_k} (\vec{n}_{k+1}^\dag ,\vec{n}_{k}) -
e^{Q_k -Q_{k-1}} (\vec{n}_{k}^\dag ,\vec{n}_{k-1}) \right)  , \\
\frac{ d Q_k }{dt} &=  - 4\nu_0 \lambda_k ,\qquad \frac{ d\vec{n}_k  }{ dt } =  \mathcal{ O}(\epsilon),
\end{aligned}\end{equation}
where again $\lambda _k=\mu _k+i\nu _k $ and the other variables are given by (\ref{eq:140.1}).
Now we have additional equations describing the evolution of the polarization vectors.
But note, that their evolution is slow, and in addition the products  $(\vec{n}_{k+1}^\dag ,\vec{n}_{k})$
multiply the exponents $e^{Q_{k+1}-Q_k}$ which are also of the order of $\epsilon$. Since we are
keeping only terms of the order of $\epsilon$ we can  replace   $(\vec{n}_{k+1}^\dag ,\vec{n}_{k})$
by their initial values
\begin{equation}\label{eq:dnk}\begin{split}
\left. (\vec{n}_{k+1}^\dag ,\vec{n}_{k}) \right|_{t=0} = m_{0k}^2 e^{2i\sigma_{k}}, \qquad k=1,\dots , N-1.
\end{split}\end{equation}

We will consider most general form of the polarization vectors:
\begin{equation}\label{eq:pol1}\begin{split}
| \vec{n}_k \rangle  & = \left(\begin{array}{c} \cos(\theta_k) e^{i\gamma_k} \\ \sin(\theta_k) e^{-i\gamma_k}   \end{array}\right), \\
 \langle \vec{n}_{k+1}^\dag | \vec{n}_k\rangle &= \cos(\theta_{k+1}\cos(\theta_{k}) e^{-i(\gamma_{k+1}-\gamma_k)} + \sin(\theta_{k+1}\sin(\theta_{k}) e^{i(\gamma_{k+1}-\gamma_k)}
= \rho_k e^{i\sigma_k}, \\
\rho_k^2 &= \cos^2(\gamma_{k+1} -\gamma_k) \cos^2(\theta_{k+1}-\theta_k)+ \sin^2(\gamma_{k+1} -\gamma_k) \cos^2(\theta_{k+1}+\theta_k). \\
 \sigma_k &= -\arctan \left( \tan (\gamma_{k+1} -\gamma_k) \frac{\cos(\theta_{k+1}+\theta_k) }{\cos(\theta_{k+1}-\theta_k)} \right), \\
 a_k&= \frac{ \nu_0}{2} \rho_k^2 \exp(-\nu_0(\xi_{k+1}-\xi_k)) \exp( -i (\delta_{k+1} -\delta_k -\sigma_k+\pi)/2).
\end{split}\end{equation}
In our previous papers we considered configurations for which $|\vec{n}_k \rangle $ are real, i.e., $\gamma_k =0$. Note that the effect of the polarization vectors could be viewed as change of the distance between the solitons and between the phases.

The  system (\ref{eq:Vnuk0}) was derived for the Manakov system $n=2$ by other methods in \cite{GKDM09}. There the GCTC model it was tested numerically and found to give very good agreement with the numerical solution of the Manakov model. However the tests were done only for real values of the polarization vectors, i.e., all $\gamma_k=0$, $k=1,\dots, N$. Below we will take into account the effect of $\gamma_k$ onto the dynamical regimes of the solitons.

\subsection{Asymptotic Regimes: General Approach}

We first briefly remind the main results concerning the CTC model \cite{PRL77, PRE55, PLA98, GDM07, GKDM09}.
The CTC is completely integrable model; it allows Lax representation $L_t =[A.L]$, where:
\begin{equation}\label{eq:Laxctc}\begin{split}
L= \sum_{s=1}^{N} \left(b_s E_{ss} +a_s(E_{s,s+1} + E_{s+1,s}) \right), \quad
A= \sum_{s=1}^{N} \left(a_s(E_{s,s+1} - E_{s+1,s}) \right),
\end{split}\end{equation}
where $a_s =\exp ((Q_{s+1}-Q_s)/2)$, $b_s =\mu_{s,t} +i\nu_{s,t} $ and the matrices
$E_{ks}$ are determined by $(E_{ks})_{pj} =\delta_{kp} \delta_{sj}$. The eigenvalues of $L$ are
integrals of motion and determine the asymptotic velocities.

The GCTC derived in  \cite{VG98, VG03, GDM07, GKDM09, BJP38} is also a completely integrable model. It allows Lax representation just like the standard real Toda chain \cite{Fla, Moser, Manakov} $\tilde{L}_t =[\tilde{A}. \tilde{L}]$, where:
\begin{equation}\label{eq:gLaxctc}\begin{split}
\tilde{L}= \sum_{s=1}^{N} \left(\tilde{b}_s E_{ss} + \tilde{a}_s(E_{s,s+1} + E_{s+1,s}) \right), \quad
A= \sum_{s=1}^{N} \left(\tilde{a}_s(E_{s,s+1} - E_{s+1,s}) \right),
\end{split}\end{equation}
where $\tilde{a}_s =m_{0s}e^{i\sigma_{s}} a_s$, $b_s =\mu_{s} +i\nu_{s} $. Like for the scalar case,
the eigenvalues of $\tilde{L}$ are integrals of motion.
If we denote by $\zeta_s = \kappa_s +i\eta_s$ (resp. $ \tilde{ \zeta}_s=\tilde{ \kappa}+i \tilde{ \eta}_s
$) the set of eigenvalues of $L$ (resp. $\tilde{L}$) then their real parts
$\kappa_s$ (resp. $\tilde{ \kappa}_s$)  determine the asymptotic velocities for the soliton train
described by CTC (resp. GCTC).
Thus, starting from the set of initial soliton parameters  we can calculate $L|_{t=0}$
(resp. $\tilde{ L}|_{t=0}$),  evaluate the real parts of their eigenvalues and thus determine
the asymptotic regime of the soliton train.

\begin{description}
  \item[Regime (i)] $\kappa _k\neq \kappa _j $ (resp. $\tilde{\kappa} _k\neq \tilde{ \kappa} _j $)
  for $k\neq j $, i.e., the asymptotic velocities are all different. Then we have asymptotically
separating, free solitons, see also \cite{Arn, PRL77, PRE55, PLA98}

  \item[Regime (ii)] $\kappa _{1} = \kappa _{2} = \dots = \kappa_{N} =0$
  (resp. $\tilde{ \kappa} _{1} = \tilde{ \kappa} _{2} = \dots = \tilde{ \kappa}_{N} =0$),
  i.e., all $N $ solitons move with the same mean asymptotic velocity,
and form a ``bound state.''

  \item[Regime (iii)] a variety of intermediate situations when one
group (or several groups) of particles move with the same mean
asymptotic velocity; then they would form one (or several) bound
state(s) and the rest of the particles will have free asymptotic
motion.
\end{description}
\begin{remark}\label{rem:1}
The sets of eigenvalues of $L$ and $\tilde{ L}$ are generically different.
Thus varying only the  polarization vectors one can change the asymptotic regime of the
soliton train.
\end{remark}

Let us consider several particular cases.

\begin{description}
  \item[Case 1] $\vec{n}_1 = \cdots = \vec{n}_N$. Since the vector $\vec{n}_1$ is normalized, then
all coefficients $m_{ok}=1$ and $\sigma_{k}=0$. Then the interactions of the vector and scalar
solitons are identical.

  \item[Case 2] $(\vec{n}_{s+1}^\dag ,\vec{n}_{s})=0$. Then the GCTC splits into two unrelated GCTC: one for
  the solitons $\{1,2,\dots,s\}$ and another for $\{s+1,s+2,\dots. N\}$. If the two sets of soliton parameters are such that
  both groups of solitons are in bound state regimes, then these two bound states.

  \item[Case 3] $\langle n_{k+1}^\dag | \vec{n}_k \rangle =m_0 e^{i\varphi_0} $ -- effective change of
distance and phases of solitons. In this case we can rewrite $\tilde{a}_s =\exp((\tilde{Q}_{s+1} -\tilde{Q}_{s})/2)$,
where:
\begin{equation}\label{eq:Qst}\begin{split}
\tilde{Q}_{s+1} - \tilde{Q}_{s} =Q_{s+1} -Q_s + \ln m_0 + i\varphi_0,
\end{split}\end{equation}
i.e., the distance between any two neighboring vector solitons has changed by $\ln (m_0/2\nu_0)$; similarly have  the phases.

\end{description}

\section{Asymptotic Regimes for $N$-Soliton Trains with $N\geq 4$}

The asymptotic regimes for scalar solitons and for small values of $N$ are   known for long time now, see \cite{PRL77, PRE55, PLA98}. Obviously for $N=2$ we have only two possibilities: BSR and FAR. For $N=3$ for the first time there appears MAR when two of the solitons form a bound state while the third one goes away off them. For $N>3$ there were only fragmentary results, see the quasi-equidistant propagation of solitons in \cite{PLA98}.

For the Manakov solitons formally the method is the same.
The idea to use the integrability of CTC in order to develop a tool for the analysis of asymptotic behavior of $N$-soliton trains was developed in \cite{VG98, GDM07, GKDM09, BJP38}. Roughly speaking we have to use the characteristic polynomial of $L_N$ whose generic form is:
\begin{equation}\label{eq:Pz}\begin{split}
P(z) = \det (L_N - z\mathbb{1}) = \sum_{k=0}^{N} p_k(\vec{a},\vec{b}) z^k = \prod_{k=1}^N (z - z_k).
\end{split}\end{equation}
Next we have to analyze the roots $z_k$ and formulate the conditions on the soliton parameters for which
\begin{equation}\label{eq:zkA}\begin{split}
\mbox{i)} \qquad \re z_k=0; \qquad \mbox{ii)} \qquad \im z_k=0;
\end{split}\end{equation}
Formally condition i) in (\ref{eq:zkA}) ensures the BSR, while condition ii) in (\ref{eq:zkA}) is responsible for the FAR.

However each soliton now has 6 parameters, so 3, 4 and 5 solitons will be parametrized by 18, 24 and 30 parameters respectively. The large number of parameters makes it difficult to derive explicit analytical results, or to do an exhaustive numerical studies. Of course some configurations of Manakov solitons behave just like the scalar ones. This happens if all $\vec{n}_k $ are equal. Naturally our aim is consider more interesting cases and demonstrate the important role that the polarization vectors play for the soliton interactions. Indeed $m_{0k}$ in (\ref{eq:dnk}) take any value from 0 to 1, i.e., they `regulate' the strength of the interaction. In particular, if the polarization vectors of two neighboring solitons are orthogonal, then they do not interact. In addition the phases $\sigma_k$ modify the phase difference of the solitons which is a substantial factor in their interaction.


Situations when we have 2, 3 and 4 solitons are easier because we can write down explicit formulae for $z_k$ in terms of the soliton parameters in the generic case. For two and three solitons most of this analysis for scalar solitons were done \cite{PRL77, PRE55, PLA98}. For bigger values of $N$ such formulas are not done even for the scalar case, in which the number of the soliton parameters are $4N$. For $N=4$ already the formulae for $z_k$ are involved; in addition the number of the parameters is $4N=16$. Therefore for $N\geq 4$ even for the scalar case only special configurations of soliton parameters are known. They are related to special choices of the soliton parameters that simplify the characteristic polynomial so that it reduces to, say a biquadratic equation.
In addition, when it comes to Manakov solitons, the number of the parameters becomes $6N$.

Our aim here will be: first to revisit the particular cases considered before and, second, to propose special soliton configurations responsible for the BSR and FAR for any number of solitons. We will illustrate our results by several figures.

\subsection{Asymptotic Regimes for Manakov Solitons}

Let us now outline some effective ways of choosing soliton parameters that would ensure
given asymptotic behavior of the solitons.
The soliton parameters of the Manakov $N$-soliton train are $6N$ and detailed study of the regions in which the solitons will develop given asymptotic regime does not seem possible.
However we will outline several ways to effectively pick up configurations ensuring BSR or FAR asymptotic regimes.

Let us also remind several important issues that one needs to consider. First we need to specify what we will consider as asymptotic state. Obviously we need a criterium that would ensure us that we are in the asymptotic region. In our case we have two scales: $\epsilon^{1/2}$ and $\epsilon$ that are fundamental for the adiabatic approximation. It is reasonable to assume that the asymptotic times must be of the order of $1/\epsilon$. Our choices of soliton parameters are such that $\epsilon \simeq 10^{-2}$. So one could expect that the asymptotic times would be of the order of $\epsilon^{-1} \simeq 100$. At the same time we extend our numerics to about $t_{\rm as} \simeq 1000$ and in most cases we find good match between the CTC prediction and the numerics of Manakov during all that period. This means that CTC models the Manakoc model much better that we can expect. We can see from the figures presented here and from many others that we have done that the match could be much better.

Indeed, let us assume that we know how to split the $30$-dimensional space of our soliton parameters into regions that correspond to the different asymptotic regimes. Obviously, if we choose the soliton parameters to be close to the `border' between two different regimes we can expect that we would have a `transition' area between the regimes, so the deviation from the CTC model will come up sooner than 1000. This is what we can see in the right panels of  Figures \ref{fig:2},  \ref{fig:3}. There for $t>300$ we see that the bound state of 5 solitons in fact transforms into a MAR. The first and the fifth solitons `peel off' and go freely away, and the other three still stay in a BSR. It seems that increasing the differences between  the amplitudes stabilizes  the BSR.

The general criterium that ensures FAR or BSR is based on the following well known proposition coming from linear algebra.

\begin{proposition}\label{thm:1}
Let $L_0$ be symmetric $L_0= L_0^T$ matrix with real-valued matrix elements. Then its eigenvalues $z_{0j}$ will be real and different, i.e., $z_{0j}\neq z_{0k}$ for $k\neq j$.
\end{proposition}

\begin{corollary}\label{cor:1}
Let $L_1$ be symmetric (not hermitian) $L_1= L_1^T$ matrix with purely imaginary matrix elements. Then its eigenvalues $z_{1j}$ will be purely imaginary and different, i.e., $z_{1j}\neq z_{1k}$ for $k\neq j$.
\end{corollary}
\begin{proof}
Follows directly from the Proposition if we consider $L_1=iL_0$.
\end{proof}

In addition below we will assume that $\nu_0=0.5$ and $\mu_0=0$.

\begin{figure}[h!]
 \centerline{\includegraphics[width= .49\textwidth]{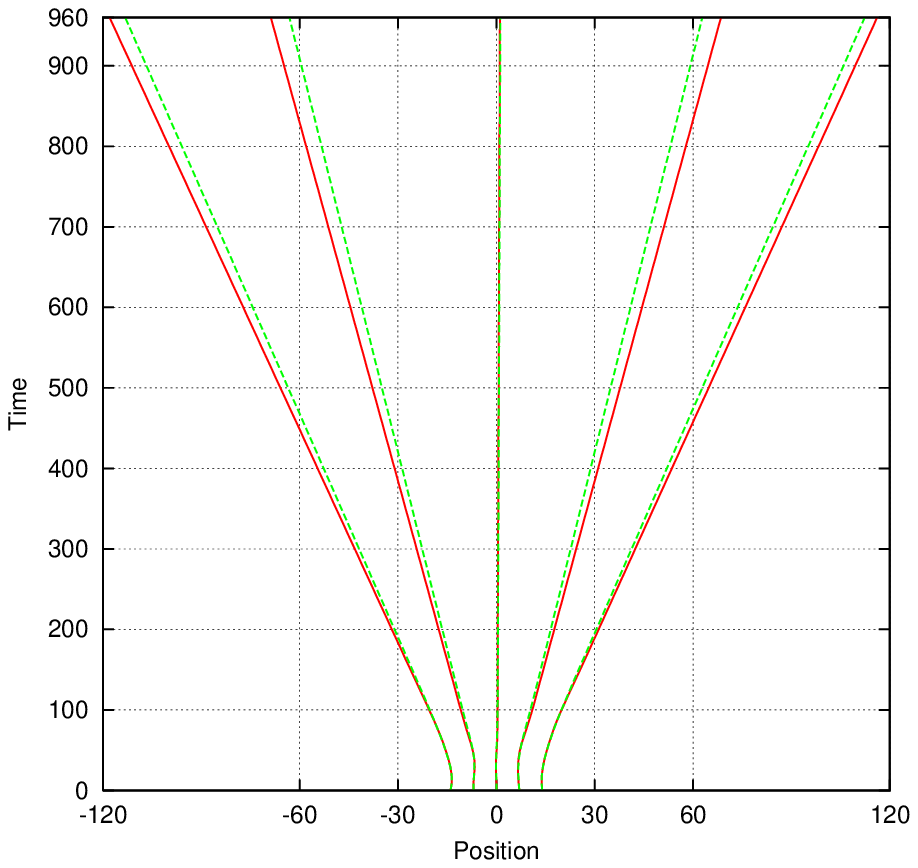}\qquad 
  \includegraphics[width= .49\textwidth]{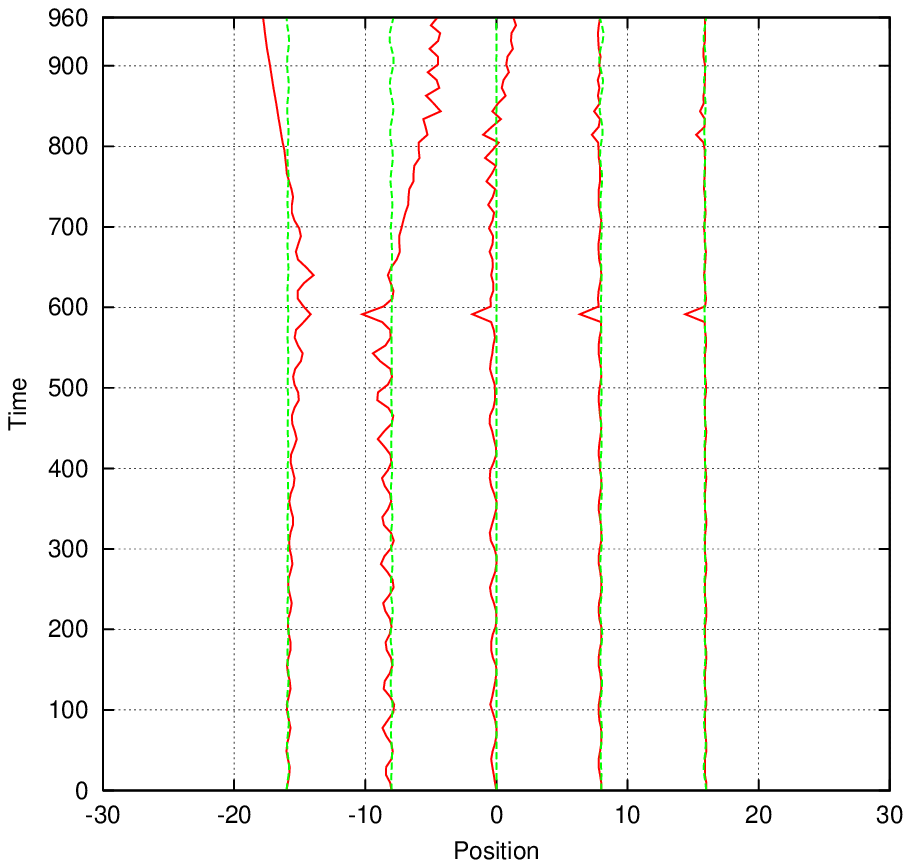}}
 \caption{
 Left panel: FAR with initial conditions  $r_0=7.0$, $\mu_{00}=0.01$, $\nu_{00}=0.0$, $g_0=9$; Right panel: BSR with initial conditions $r_0=8.0$, $\mu_{00}=0.0$, $\nu_{00}=0.05$, $g_0=9$. The rest of the parameters are defined by eqs. (\ref{eq:FAR2}) and (\ref{eq:BSR2}) respectively.}
\label{fig:1}
  \end{figure}

\subsection{Generic FAR Configurations}
These configurations are characteristic for the real Toda chain solved by Moser \cite{Moser, Manakov, Fla}.

In what follows  we choose the polarization vectors $\vec{n}_k$ by setting:
\begin{equation}\label{eq:pol}\begin{split}
 \theta_k = \frac{ k\pi}{13}, \qquad \gamma_k= \frac{k\pi}{g_0}.
\end{split}\end{equation}
where $g_0 =8$, or $g_0 =9$.

For the CTC using the Proposition we obtain:
\begin{equation}\label{eq:FAR}\begin{split}
\im b_k|_{t=0} =0, \qquad \im a_k|_{t=0} = 0,
\end{split}\end{equation}
which means that
\begin{equation}\label{eq:FAR2}\begin{aligned}
\nu_k|_{t=0} &=0.5, \qquad  b_k|_{t=0} = \mu_k|_{t=0} = \mu_{0k}, \qquad \theta_k = \frac{k\pi }{13}, \qquad  \gamma_k = \frac{k\pi}{g_0}, \\
 \xi_{0k} &= (k-3) r_0, \qquad \mu_{0k} = (k-3) \mu_{00}, \quad \nu_{0k} = 0.5 + (k-3) \nu_{00},\\
\delta_{0,1}&=0, \qquad  \delta_{0,k+1} - \delta_{0,k}= \sigma_k,
\end{aligned}\end{equation}

Indeed, from the Proposition the eigenvalues of $L$ will be real and different, which is FAR.
A particular case of (\ref{eq:FAR2}) as configuration ensuring FAR for scalar solitons was noticed long ago, namely choosing solitons with equal amplitudes (i.e., $\Delta \nu_k=0$) and and out-of phase $\delta_{k+1}-\delta_k=\pi$  \cite{Arn}.
However, eq. (\ref{eq:FAR2}) provides more general configurations, in which the solitons may have non-vanishing initial velocities,  see Figure \ref{fig:1}.

\begin{figure}[h!]
  \centerline{\includegraphics[width= .49\textwidth]{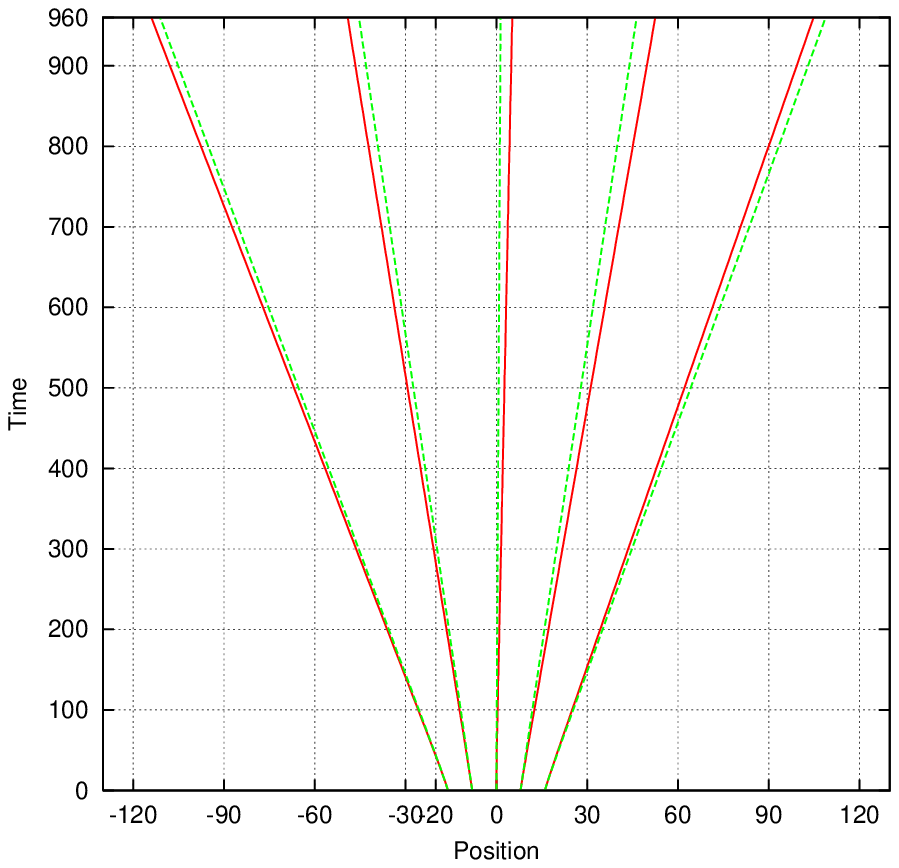}\qquad 
  \includegraphics[width= .49\textwidth]{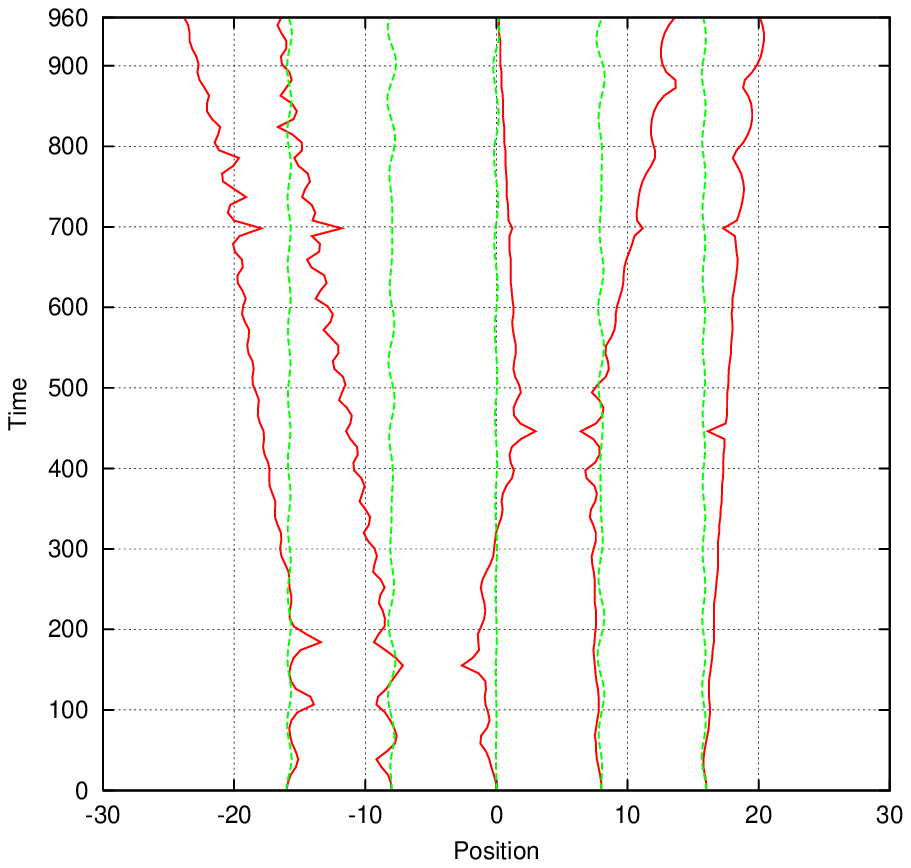}}
 \caption{  Left panel: FAR with initial conditions  $r_0=8.0$, $\mu_{00}=0.02$, $\nu_{00}=0.0$, $g_0=4$; Right panel: BSR with initial conditions $r_0=8.0$, $\mu_{00}=0.0$, $\nu_{00}=0.03$, $g_0=4$. The rest of the parameters are defined by eqs. (\ref{eq:FAR2}) and (\ref{eq:BSR2}) respectively.}
 \label{fig:2}
  \end{figure}

\subsection{Generic BSR Configurations}
Here we use the Corollary and impose on $L$ the conditions:
\begin{equation}\label{eq:BSR}\begin{split}
\re b_k|_{t=0} =0, \qquad \re a_k|_{t=0} = 0,
\end{split}\end{equation}
which means that
\begin{equation}\label{eq:BSR2}\begin{aligned}
b_k|_{t=0} &= i\nu_k|_{t=0} = i\nu_{0k}, \qquad \theta_k = \frac{k\pi }{13}, \qquad  \gamma_k = \frac{k\pi}{g_0}, \\
 \xi_{0k} &= (k-3) r_0, \qquad \mu_{0k} = 0.0, \qquad \nu_{0k} = 0.5 + (k-3) \nu_{00},\\
\delta_{0,1}&=0, \qquad  \delta_{0,k+1} - \delta_{0,k}= \sigma_k +\pi,
\end{aligned}\end{equation}
This is also rather general and simple condition on the soliton parameters that fixes the
initial velocities to be 0, but does not put restrictions (except the adiabatic ones) on the amplitudes and on the initial positions of the solitons.

\subsection{Symmetric Configurations of Soliton Parameters}
In addition to these we find other configurations of soliton parameters that provide FAR or BSR. To this end we use special symmetric constraints on $L$ described below. These constraints will leave only one of $\nu_{0k}$ and $a_{0k}$ independent. As a result the characteristic polynomial of $L$ will factorize and we will find that all roots are proportional to each other.

Let us give few examples of them. We will provide the corresponding Lax matrix, its characteristic polynomial and eigenvalues.

\begin{itemize}
  \item
$N=3$, $P_3 = z (z^2 - 4(a^2 + b^2))$:
  \begin{equation}\label{eq:L3} \begin{split}
L_3 &= \left(\begin{array}{ccc} b & \sqrt{2}a & 0 \\ \sqrt{2}a & 0 & \sqrt{2}a \\ 0 & \sqrt{2}a & -b  \end{array}\right),   \\
z_{1,2} &= \pm 2\sqrt{a^2 +b^2}, \qquad z_3=0;
    \end{split}\end{equation}

\item $N=4$, $P_4 =  (z^2 - a^2 - b^2) (z^2 - 9(a^2 + b^2))$
  \begin{equation}\label{eq:L4}\begin{split}
  L_4 &= \left(\begin{array}{cccc} 3b & \sqrt{3} a & 0 & 0 \\ \sqrt{3} a & b & 2a &0 \\ 0 & 2a & -b & \sqrt{3} a \\ 0 & 0 & \sqrt{3} a & -3b  \end{array}\right),
  \\ z_{1,2} &= \pm \sqrt{a^2 +b^2}, \qquad z_{3,4}  = \pm 3\sqrt{a^2 +b^2};
  \end{split}\end{equation}

 \item $N=5$, $P_5 =  z (z^2 - a^2 - b^2) (z^2 - 4(a^2 + b^2))$
\begin{equation}\label{eq:L5}  \begin{split}
  L_5 &= \left(\begin{array}{ccccc} 2b & \sqrt{3} a & 0 & 0 & 0 \\ \sqrt{2} a & b & 2 a &0 & 0 \\ 0 & 2 a & 0 & \sqrt{3} a & 0 \\ 0 & 0 & \sqrt{3} a & -b & \sqrt{2}a \\
  0 & 0 & 0 & \sqrt{2} a & -2b \end{array}\right)   \\
  z_{1,2} &= \pm \sqrt{a^2 +b^2}, \qquad z_{3,4} = \pm 2\sqrt{a^2 +b^2}, \qquad z_5=0;
 \end{split}\end{equation}

\begin{figure}[h!]
 \centerline{\includegraphics[width= .49\textwidth]{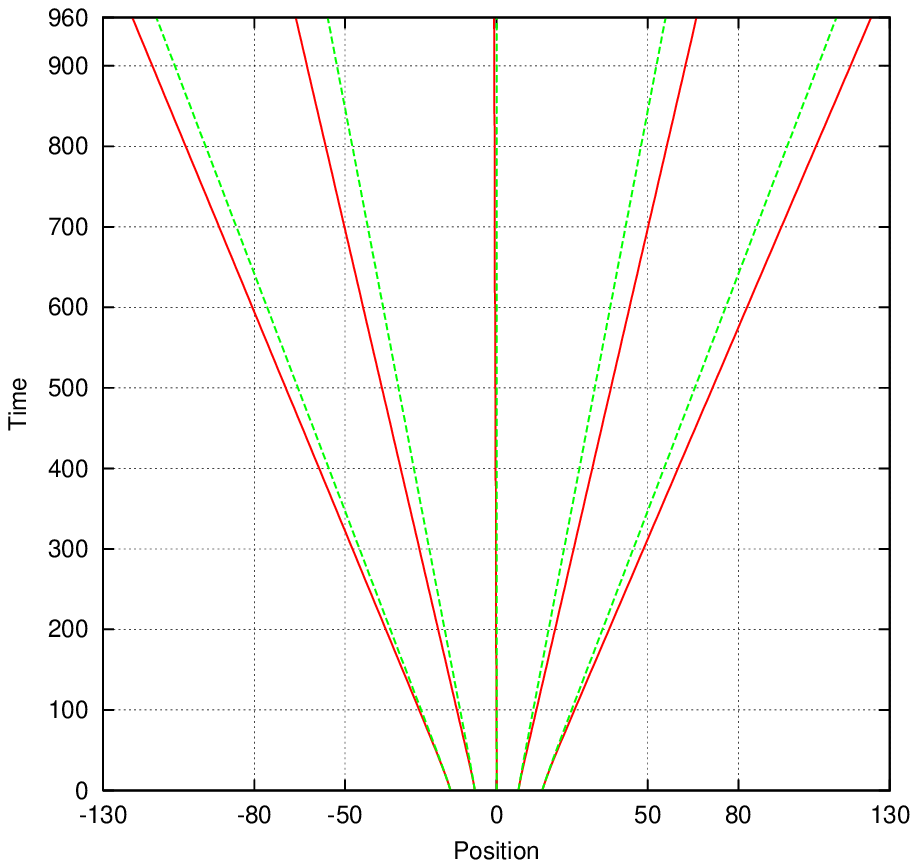}\qquad 
  \includegraphics[width= .49\textwidth]{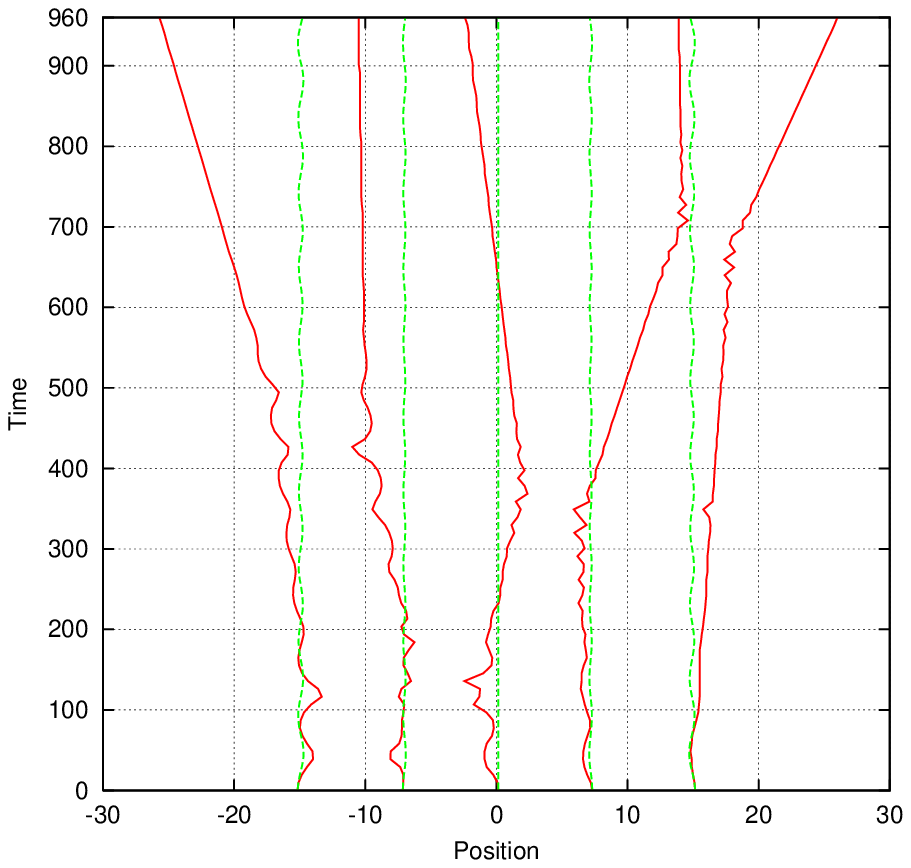}}
 \caption{
 Left panel: FAR with initial conditions  $\mu_{00}=0.02$, $\nu_{00}=0.0$, $g_0=4$; Right panel: BSR with initial conditions $\mu_{00}=0.0$, $\nu_{00}=0.03$, $g_0=4$. The rest of the parameters are defined by eqs. (\ref{eq:far1}) and (\ref{eq:bsr1}) respectively.}
  \label{fig:3}
  \end{figure}

 \item $N=6$, $P_6 =   (z^2 - a^2 - b^2) (z^2 - 9(a^2 + b^2)) (z^2 - 25(a^2 + b^2))$:
 \begin{equation}\label{eq:L6}\begin{split}
  L_6 &= \left(\begin{array}{cccccc} 5b & \sqrt{5} a & 0 & 0 & 0 &0 \\ \sqrt{5} a & 3b & \sqrt{3} a &0 & 0 & 0 \\ 0 & \sqrt{8}a & b & 3 a & 0 & 0\\ 0 & 0 & 3 a & -b & \sqrt{8}a & 0 \\ 0 & 0 & 0 & \sqrt{8} a & -3b  & \sqrt{5} a \\  0 & 0 & 0 & 0 & \sqrt{5} a & -5b  \end{array}\right),  \\
  z_{1,2} &= \pm \sqrt{a^2 +b^2}, \qquad z_{3,4} = \pm 3\sqrt{a^2 +b^2}, \qquad z_{5,6}=\pm 5\sqrt{a^2 +b^2}.
 \end{split}\end{equation}

\end{itemize}
Such examples can be found for any value of $N$; from algebraic point of view they are related to the the maximal embedding of $sl(2)$ as a subalgebra of $sl(N)$.

In order to ensure FAR or BSR we need to impose on $a$ and $b$ the condition that
\begin{equation}\label{eq:cond1}\begin{split}
 \mbox{FAR} \qquad a^2 +b^2 >0, \qquad  \mbox{BSR} \qquad a^2 +b^2 <0.
\end{split}\end{equation}

Initial conditions for BSR of 5 scalar solitons:
\begin{equation}\label{eq:BSR1}\begin{aligned}
 \xi_1 &=-2r_0 +\frac{\ln 6}{2\nu_0}, &\;  \xi_2 &=-r_0 +\frac{\ln 3}{2\nu_0}, &\; \xi_3 &=0, &\;  \xi_4 &=r_0 -\frac{\ln 3}{2\nu_0}, &\;  \xi_5 &=2r_0 -\frac{\ln 6}{2\nu_0},\\
 \nu_k &= 0.5 + (3-k)\nu_{00}, &\;  \mu_k &=0, &\; \delta_{k}&= k\pi, &\quad k&=1,\dots ,5.
\end{aligned}\end{equation}

Initial conditions for FAR of 5 scalar solitons:
\begin{equation}\label{eq:xi0kS}\begin{aligned}
 \xi_1 &=-2r_0 +\frac{\ln 6}{2\nu_0}, &\;  \xi_2 &=-r_0 +\frac{\ln 3}{2\nu_0}, &\; \xi_3 &=0, &\;  \xi_4 &=r_0 -\frac{\ln 3}{2\nu_0}, &\;  \xi_5 &=2r_0 -\frac{\ln 6}{2\nu_0},\\
 \nu_k &= 0.5 , &\;  \mu_k &=(3-k)\mu_{00}, &\; \delta_{k}& =\frac{k\pi}{2}, &\quad k&=1,\dots ,5.
\end{aligned}\end{equation}

For Manakov solitons the initial positions  are determined by:
\begin{equation}\label{eq:xi0kM}\begin{aligned}
\xi_{10} &= -2r_0 - \frac{1}{2\nu_0} \ln \frac{ m_{01}m_{02}m_{03}m_{04}}{6}, &\;
\xi_{20} &= -r_0 - \frac{1}{2\nu_0} \ln \frac{ m_{02}m_{03}m_{04}}{3m_{01}}, \\
\xi_{30} &=  -\frac{1}{2\nu_0} \ln \frac{ m_{03}m_{04}}{m_{01}m_{02}}, \\
\xi_{40} &= r_0 +\frac{1}{2\nu_0} \ln \frac{ m_{01}m_{02}m_{03}}{3m_{04}}, &\;
\xi_{50} &= 2r_0 + \frac{1}{2\nu_0} \ln \frac{ m_{01}m_{02}m_{03}m_{04}}{6},
\end{aligned}\end{equation}

For the numerics we again fix the polarization vectors as in (\ref{eq:pol}) and evaluate $\xi_{0k}$ by the formula (\ref{eq:xi0kM}). The result is:
\begin{equation}\label{eq:inpos}\begin{split}
\xi_{01} = -15..., \qquad  \xi_{02} = -9...
\end{split}\end{equation}

In order to have FAR we choose the amplitudes, velocities and the phases of the solitons by:
\begin{equation}\label{eq:far1}\begin{aligned}
\nu_k &= 0.5, &\; \mu_k &=(k-3)\mu_{00}, &\; k&=1,2,\dots, 5 \\
\delta_{10} &=0, \quad \delta_{20} = \delta_{10}+\sigma_1+\pi, &\quad
\delta_{30} &=  \delta_{10}+\sigma_{1}+\sigma_{2}+\pi, \\
\delta_{40} &=   \delta_{30}+\sigma_{1}+\sigma_{2}+\sigma_{3}+\pi, &\;
\delta_{50} &= \delta_{40}+\sigma_{1}+\sigma_{2}+\sigma_{3}+\sigma_{4}+\pi,
\end{aligned}\end{equation}

For the BSR we choose the amplitudes, velocities and the phases of the solitons by:
\begin{equation}\label{eq:bsr1}\begin{aligned}
\nu_k &= 0.5 +(k-3)\nu_{00}, &\; \mu_k &=0, &\; k&=1,2,\dots, 5\\
\delta_{10} &=0, \quad \delta_{20} = \delta_{10}+\sigma_1, &\quad
\delta_{30} &=  \delta_{10}+\sigma_{1}+\sigma_{2}, \\
\delta_{40} &=   \delta_{30}+\sigma_{1}+\sigma_{2}+\sigma_{3}, &\;
\delta_{50} &= \delta_{40}+\sigma_{1}+\sigma_{2}+\sigma_{3}+\sigma_{4},
\end{aligned}\end{equation}

\subsection{Numeric Values for the Initial Parameters}

In the Tables  we list the numeric values for $m_{0k}$ and $\sigma_k$ for the two typical choices of $\theta_k$ and $\gamma_k$ used above.

\begin{table}
\begin{center}
  \begin{tabular}{|l|r|r|}
    \hline
    $\delta_{0k}$ & left panel & right panel \\ \hline
    $k=1$ &  0.0 & 0.0 \\
    $k=2$ & 2.868037 & -0.273554 \\
    $k=3$ & -0.405708 & -0.405708 \\
    $k=4$ & 2.781038  & -0.360554 \\
    $k=5$ &  -0.150741 & -0.150741 \\
    \hline
  \end{tabular}
\qquad
  \begin{tabular}{|l|r|r|}
    \hline
    $\delta_{0k}$ & left panel & right panel \\ \hline
    $k=1$ &  0.0 & 0.0 \\
    $k=2$ & 2.484841 & -0.656751 \\
    $k=3$ & -1.006917 & -1.006917 \\
    $k=4$ &  2.258187 &-0.883405 \\
    $k=5$ & -0.354039 &  -0.354039 \\
    \hline
  \end{tabular}
  \caption{Initial phases for Fig. \ref{fig:1} and Fig. \ref{fig:2}}\label{tab:2}
\end{center}
\end{table}

\begin{table}
  \centering
  \begin{tabular}{|l|r|r|r|r|}
    \hline
& \multicolumn{2}{|c|}{left panel} & \multicolumn{2}{|c|}{right panel} \\ \hline
   & $\delta_{0k}$ &  $\xi_{0k}$ & $\delta_{0k}$ &  $\xi_{0k}$  \\ \hline
    $k=1$ &  0.0 & -15.154654 & 0.0 & -15.154654 \\
    $k=2$ & 2.484841 & -7.133487 & -0.656751 &  -7.133487 \\
    $k=3$ & -1.006917 & .140982 & -1.006917 & 0.140982 \\
    $k=4$ & 2.258187 &  7.305540 & -0.883405 & 7.305540 \\
    $k=5$ &  -.354039 & 15.154654 & -0.354039 & 15.154654\\
    \hline
  \end{tabular}
  \caption{Initial phases and positions for Fig. \ref{fig:3}}\label{tab:3}
\end{table}

\section*{Conclusions and Discussion}
The above analysis can be extended to any number of solitons. As we mentioned above, the symmetric Lax matrices are realizations of the maximal embedding of the $sl(2)$ algebra as a subalgebra of $sl(N)$. In this case we effectively reduce the $N$-soliton interactions to an effective $2$-soliton interactions. Therefore the symmetric configurations studied above allow only two asymptotic regimes: BSR and FAR. We make the hypothesis that it would be possible to construct more general symmetric Lax matrices that would be responsible for effective $3$-soliton interactions.
In this paper we included numerical tests only for 3 soliton interactions. However previously we have run test starting with 2-solitons and ending with 9-soliton configurations. Our results are that the CTC models adequately not only the purely solitonic interactions, but also the effects of external potentials and other perturbations on them.

An interesting question is how long should we wait for the asymptotic regime. This question is directly related to the other one: What are the limits of applicability of CTC? In our simulations we have chosen  $\varepsilon_0 \simeq 0.01$ which means that the asymptotic time must be of the order of $1/\varepsilon_0 \simeq 100$. At the same time in a number of cases we find good match between the CTC and the numeric solutions of Manakov model even until $1\;000$. This is what we see in our tests in this paper for the free asymptotic regimes (right panels of all Figures). The situation is different for the bound state regimes. While in Figure \ref{fig:1} we see good match until about 700, in Figures \ref{fig:1} and \ref{fig:3} the good match goes until 300. After that the trajectories of CTC keep to the BSR, but some of the real solitons `escape aeay` after that.
However in all cases we find that CTC provides good descriptions until times about three times larger than the asymptotic one.

\section*{Acknowledgements}\label{sec:Ack}
MDT was supported by Fulbright -- Bulgarian-American Commission for Educational Exchange under Grant No 19-21-07.

\end{document}